\documentclass{article}

\usepackage[utf8]{inputenc}
\usepackage[T1]{fontenc}
\usepackage{libertine}
\usepackage[mathscr]{eucal}  
\usepackage{inconsolata}
\usepackage{microtype}
\usepackage{algorithm}
\usepackage[noend]{algpseudocode}
\usepackage{amsmath, amssymb, amsthm}
\usepackage{bm}
\usepackage{bbm}
\usepackage{graphicx}
\usepackage{thmtools}
\usepackage{thm-restate}
\usepackage{hyperref} 
\usepackage[capitalize, nameinlink]{cleveref}
\usepackage{comment}
\usepackage{enumerate}
\usepackage{fullpage}
\usepackage{mathtools}
\usepackage{subcaption}
\usepackage[dvipsnames]{xcolor}
\usepackage{nag}
\usepackage{stmaryrd}
\usepackage{upgreek}
\usepackage[round]{natbib}
\usepackage{tcolorbox}
\usepackage{booktabs}
\usepackage{authblk}

\graphicspath{{figures/}{.}}

\newcommand{\declarecolor}[2]{\definecolor{#1}{RGB}{#2}\expandafter\newcommand\csname #1\endcsname[1]{\textcolor{#1}{##1}}}
\declarecolor{White}{255, 255, 255}
\declarecolor{Black}{0, 0, 0}
\declarecolor{LightGray}{216, 216, 216}
\declarecolor{Gray}{127, 127, 127}
\declarecolor{Orange}{237, 125, 49}
\declarecolor{LightOrange}{251,229, 214}
\declarecolor{Yellow}{255, 192, 0}
\declarecolor{LightYellow}{255, 242, 200}
\declarecolor{Blue}{91, 155, 213}
\declarecolor{LightBlue}{222, 235, 247}
\declarecolor{Green}{112, 173, 71}
\declarecolor{LightGreen}{226, 240, 217}
\declarecolor{Navy}{68, 114, 196}
\declarecolor{LightNavy}{218, 227, 243}

\hypersetup{
	colorlinks=true,
	pdfpagemode=UseNone,
	citecolor=Navy,
	linkcolor=Navy,
	urlcolor=Navy,
}

\crefformat{equation}{Eq.~(#2#1#3)}

\DeclareMathOperator*{\defeq}{\overset{def}{=}}
\newcommand{\mat}[1]{\mathbf{#1}}
\newcommand{\tensor}[1]{\bm{\mathscr{#1}}}
\newcommand{\frobenius}{\textnormal{F}}
\newcommand{\STHOSVD}{\textnormal{ST-HOSVD}}
\newcommand{\HOSVD}{\textnormal{HOSVD}}
\newcommand{\HOOI}{\textnormal{HOOI}}

\newcommand{\R}{\mathbb{R}}

\renewcommand{\epsilon}{\varepsilon}

\DeclarePairedDelimiter{\inner}{\langle}{\rangle}

\DeclarePairedDelimiter{\set}{\{}{\}}
\DeclarePairedDelimiter{\parens}{(}{)}

\DeclarePairedDelimiter{\norm}{\lVert}{\rVert}

\theoremstyle{plain}
\newtheorem{theorem}{Theorem}[section]
\newtheorem{lemma}[theorem]{Lemma}

\theoremstyle{definition}
\newtheorem{definition}[theorem]{Definition}

\newtheorem{remark}[theorem]{Remark}

\newtheorem{fact}[theorem]{Fact}

\title{A Tight Lower Bound for the Approximation Guarantee of Higher-Order Singular Value Decomposition}

\author[1]{Matthew Fahrbach}
\author[2]{Mehrdad Ghadiri}

\affil[1]{Google Research, \texttt{fahrbach@google.com}}
\affil[1]{Massachusetts Institute of Technology, \texttt{mehrdadg@mit.edu}}

\begin{document}

\date{}
\maketitle

\begin{abstract}
We prove that the classic approximation guarantee for the higher-order singular value decomposition (HOSVD) is tight
by constructing a tensor for which HOSVD achieves an approximation ratio of $N/(1+\varepsilon)$, for any $\varepsilon > 0$.
This matches the upper bound of \citet{de2000multilinear}
and shows that the approximation ratio of HOSVD cannot be improved.
Using a more advanced construction, we also prove that
the approximation guarantees for
the ST-HOSVD algorithm of \citet{vannieuwenhoven2012new}
and higher-order orthogonal iteration (HOOI) of \citet{de2000zest}
are tight by showing that they can achieve their worst-case approximation ratio of
$N / (1 + \varepsilon)$, for any $\varepsilon > 0$.
\end{abstract}

\tableofcontents

\newpage
\section{Introduction}
\label{sec:introduction}

Tensor decomposition, the higher-order generalization of matrix factorization, is a powerful framework for analyzing multiway data, with countless applications in data mining, machine learning, and signal processing~\citep{kolda2009tensor,rabanser2017introduction,sidiropoulos2017tensor,jang2021fast}.
The two most widely used decompositions are the CP and Tucker decompositions.
A fundamental difference between matrices and tensors is that computing the rank of a tensor is NP-hard~\citep{hillar2013most}.
Consequently, tensor decomposition algorithms often require the rank to be specified in advance, after which they iteratively optimize the low-rank factors to best fit the data.
Despite the simplicity of this approach, it is highly effective in practice across many domains.

Given a tensor $\tensor{X} \in \R^{I_1 \times I_2 \times \dots \times I_N}$
and multilinear rank $\mat{r} = (R_1,R_2,\dots,R_n)$,
the standard objective for fitting a rank-$\mat{r}$ Tucker decomposition
is to minimize the least-squares reconstruction error
\begin{align*}
    L(\tensor{X},\mat{r})
    = &
    \min_{\tensor{G}, \mat{A}^{(1)}, \dots, \mat{A}^{(N)}}
    \norm{\tensor{X} - \tensor{G} \times_1 \mat{A}^{(1)} \times_2 \dots \times_N \mat{A}^{(N)}}_{\frobenius}^2
    \\ & \text{subject to }
    ~\tensor{G} \in \R^{R_1 \times R_2 \times \dots \times R_N},
    ~\mat{A}^{(n)} \in \R^{I_n \times R_n},
    ~\forall n \in [N].
\end{align*}
where $\times_{n}$ is the mode-$n$ product (defined in \Cref{subsec:tensor_products}).
Although this problem is highly nonconvex, the celebrated
\emph{higher-order singular value decomposition} (HOSVD) of \citet{de2000multilinear}
gives a Tucker decomposition with provable approximation guarantees.
To illustrate its importance, HOSVD is the first Tucker decomposition algorithm listed in the MATLAB Tensor Toolbox~\citep{matlab}.

The HOSVD computes a Tucker decomposition of a tensor $\tensor{X}$ via a simple process:
For each mode $n$, it computes the SVD of the mode-$n$ unfolding $\mat{X}_{(n)} \in \R^{I_n \times I_1 \cdots I_{n-1} I_{n+1} \cdots I_N}$ and sets the factor matrix $\mat{A}^{(n)} \in \R^{I_n \times R_N}$ to be the top-$R_n$ left singular vectors.
Then it sets the core tensor to be
$\tensor{G} \gets \tensor{X} \times_{1} {\mat{A}^{(1)}}^\intercal \times_{2} {\mat{A}^{(2)}}^\intercal \times_{3} \dots \times_{N} {\mat{A}^{(N)}}^\intercal$,
which is optimal given the fixed factors.
HOSVD comes with a strong approximation guarantee for its reconstruction error,
shown in \citet[Property 10]{de2000multilinear}
and \citet[Theorem 10.2]{hackbusch2019tensor}.

\begin{theorem}[HOSVD upper bound, informal]
\label{thm:hosvd_intro}
For any tensor $\tensor{X} \in \R^{I_1 \times I_2 \times \dots \times I_N}$
and rank $\mat{r} = (R_1,R_2,\dots,R_N)$,
\[
    \norm{\tensor{X} - \widehat{\tensor{X}}_{\HOSVD(\mat{r})}}_{\frobenius}^2
    \le
    N \cdot L(\tensor{X}, \mat{r}).
\]
\end{theorem}

In contrast to low-rank matrix approximation, where the Eckart--Young--Mirsky theorem guarantees the optimality of the truncated SVD, HOSVD does not necessarily give the best Tucker decomposition subject to the rank constraints.
For instance, \citet{kolda2009tensor} state
\emph{``the truncated HOSVD is not optimal in terms of giving the best fit as measured by the norm of the difference, but it is a good starting point for an iterative ALS algorithm.''}

Alternating least squares (ALS) methods are widely used for refining a Tucker decomposition because the objective function simplifies to an ordinary least-squares regression problem when all but one component are held fixed, e.g., a single factor matrix $\mat{A}^{(1)}$.
The most popular example is the \emph{higher-order orthgonal iteration} (HOOI) algorithm of \cite{de2000zest},
which initializes its solution with HOSVD and iteratively updates the factor matrices to monotonically decrease the reconstruction error.

Despite the long history and practical success of HOSVD and HOOI, their theoretical properties remain surprisingly elusive.
The motivation for our work is well-summarized by \citet{zhang2018tensor}:
\begin{quote}
    \emph{``HOSVD and HOOI have been widely studied in the literature. However as far as we know, many theoretical properties of these procedures, such as the error bound and the necessary iteration times, still remain unclear.''}
\end{quote}

\subsection{Our Contribution}

We show that all known methods for computing a rank-$\mat{r}$ Tucker decomposition
that have an approximation guarantee
have a worst-case approximation ratio of at least $N/(1+\epsilon)$, for any $\epsilon > 0$.
To start, we prove that the upper bound on the approximation ratio of HOSVD given in \cite{de2000multilinear} is tight.

\begin{restatable}[HOSVD lower bound]{theorem}{HosvdLowerBound}
\label{thm:hosvd_lower_bound}
For any $N \ge 2$ and $\varepsilon > 0$,
there is a tensor $\tensor{X}$ and multilinear rank $\mat{r}$ such that
\[
  \norm{\tensor{X} - \widehat{\tensor{X}}_{\HOSVD(\mat{r})}}_{\frobenius}^2
  \ge
  \frac{N}{1 + \varepsilon} \cdot L(\tensor{X}, \mat{r}).
\]
\end{restatable}

\noindent
To the best of our knowledge, our lower bound for HOSVD is novel.
We found the following lower bounds for $N = 3$ in the literature:
\begin{itemize}
    \item \citet[Example 5]{de2000multilinear}
    gave a $3 \times 3 \times 3$ tensor and set $\mat{r} = (2,2,2)$ to show that the approximation ratio is at least $(1.0880 / 1.0848)^2 = 1.0059$.
    \item \citet[Section 6.3]{vannieuwenhoven2012new}
    gave a different $3 \times 3 \times 3$ tensor and set $\mat{r}=(2,2,2)$
    to improve the approximation ratio to $(8.8188 / 7.4497)^2 = 1.4013$.
\end{itemize}
Even for $N = 3$, there is a large gap between these lower bounds
and the upper bound in \Cref{thm:hosvd_intro}.

To close this gap, we show that the HOSVD approximation ratio is at least
$N / (1 + \varepsilon)$, for any $N \ge 2$ and $\varepsilon > 0$.
Our proof is constructive: we design an adversarial symmetric tensor $\tensor{X} \in \R^{3 \times 3 \times \dots \times 3}$ for rank $\mat{r} = (2,2,\dots,2)$.
This tensor exploits the greedy nature of HOSVD by forcing it to make the same suboptimal decision in each mode, which gives a reconstruction error of $N$.
We then show that there is a rank-$\mat{r}$ Tucker decomposition
whose reconstruction error is only $1 + \varepsilon$.
The ratio of these two values proves the lower bound.

We then turn to adaptive methods---namely, \emph{sequentially truncated HOSVD} (ST-HOSVD) in \citet{vannieuwenhoven2012new} and HOOI. Analyzing approximation lower bounds for these methods is more challenging due to their adaptive behavior across the modes of the tensor and previous factor matrix decisions. For example, the tensor we construct to prove the approximation lower bound for HOSVD only yields a lower bound of $(N-1)/(1+\epsilon)$ for ST-HOSVD and HOOI.
However, by using a more intricate construction, we establish the same lower bound of $N/(1+\epsilon)$ for these methods as well.

\begin{restatable}[ST-HOSVD lower bound]{theorem}{SthosvdLowerBound}
\label{thm:st-hosvd-lower-bound}
For any $N\geq 3$ and $\epsilon>0$,
there is a tensor $\tensor{X}$ and multilinear rank $\mat{r}$ such that
\[
  \norm{\tensor{X} - \widehat{\tensor{X}}_{\STHOSVD(\mat{r})}}_{\frobenius}^2
  \ge
  \frac{N}{1 + \varepsilon} \cdot L(\tensor{X}, \mat{r}).
\]
\end{restatable}

\begin{restatable}[HOOI lower bound]{corollary}{HooiLowerBound}
\label{thm:hooi_lower_bound}
For any $N \ge 3$ and $\varepsilon > 0$,
there is a tensor $\tensor{X}$ and multilinear rank $\mat{r}$ such that
\[
  \norm{\tensor{X} - \widehat{\tensor{X}}_{\HOOI(\mat{r})}}_{\frobenius}^2
  \ge
  \frac{N}{1 + \varepsilon} \cdot L(\tensor{X}, \mat{r}).
\]
\end{restatable}

\subsection{Related Work}

In their seminal work,
\citet{de2000multilinear} introduced the HOSVD for tensors
and gave a strong approximation guarantee on the reconstruction error if we have a multilinear rank constraint.
In a follow-up work, \citet{de2000zest} studied properties of the best rank-$(1,1,\dots,1)$ decomposition and introduced the higher-order orthogonal iteration (HOOI), which uses HOSVD for initialization.
\citet{vannieuwenhoven2012new} proposed an adaptive variant of HOSVD called sequentially truncated HOSVD (ST-HOSVD), which computes each factor matrix $\mat{A}^{(n)}$ based on the current values of $\mat{A}^{(1)}, \mat{A}^{(2)}, \dots, \mat{A}^{(n-1)}$.

\citet{xu2018convergence} studied the convergence rate of HOOI via a surrogate greedy version that updates the factor matrix in each step that is closest to the current iterate.
They give conditions under which greedy HOOI (and hence HOOI) converge to a globally optimal solution.
\citet{ghadiri2023approximately} showed how to provably set the multilinear rank~$\mat{r}$
given a size constraint on the Tucker decomposition
via a connection to higher-order singular values and a knapsack problem with a matroid constraint.

There have also been many works studying continuous methods and the loss landscape for Tucker decompositions.
\citet{elden2009newton} gave an approximate Newton method for computing a rank-$(R_1,R_2,R_3)$ Tucker decomposition on a product of Grassman manifolds.
\citet{elden2011perturbation} then derived expressions for the gradient and the Hessian for this method, along with conditions for a local optimum and a first-order perturbation analysis.
\citet{frandsen2022optimization} characterized the optimization landscape of low-rank Tucker decomposition and showed that if a tensor has an exact Tucker decomposition, with an appropriate regularizer, all local minima are globally optimal.

\section{Preliminaries}
\label{sec:preliminaries}

We follow the notation and setup in the tensor decomposition survey by \citet{kolda2009tensor}.

\subsection{Notation}

The \emph{order} of a tensor is the number of dimensions, also known as ways or modes.
We denote vectors by boldface lowercase letters $\mat{x} \in \R^n$,
matrices by boldface uppercase letters $\mat{X} \in \R^{m \times n}$,
and higher-order tensors by boldface script letters $\tensor{X} \in \R^{I_1 \times I_2 \times \dots \times I_N}$.
We use normal uppercase letters for the size of an index set, e.g., $[N] = \{1,2,\dots,N\}$.
We denote the $i$-th entry of vector $\mat{x}$ by $x_i$,
the $(i,j)$-th entry of matrix $\mat{X}$ by $x_{ij}$,
and the $(i,j,k)$-th entry of a third-order tensor $\tensor{X}$ by $x_{ijk}$.

The \emph{fibers} of a tensor are the higher-order version of matrix rows and columns,
formed by fixing all indices but one.
For example, a third-order tensor has column, row, and tube fibers denoted by $\mat{x}_{:jk}$, $\mat{x}_{i:k}$, and $\mat{x}_{ij:}$, respectively.
The \emph{mode-$n$ unfolding} of a tensor $\tensor{X} \in \R^{I_1 \times I_2 \times \dots \times I_N}$ is the matrix $\mat{X}_{(n)} \in \R^{I_n \times (I_1 \cdots I_{n-1} I_{n+1} \cdots I_N)}$ that arranges the mode-$n$ fibers of $\tensor{X}$ as columns of $\mat{X}_{(n)}$, ordered lexicographically by index.

A tensor is cubical if all modes are the same size, i.e., $\tensor{X} \in \R^{I \times I \times \dots \times I}$.
A cubical tensor is \emph{symmetric} if its elements remain constant under any permutation of the indices.
For example, the third-order tensor $\tensor{X} \in \R^{I \times I \times I}$ is symmetric if,
for all $i,j,k \in [I]$,
\[
    x_{ijk}
    =
    x_{ikj}
    =
    x_{jik}
    =
    x_{jki}
    =
    x_{kij}
    =
    x_{kji}.
\]

\subsection{Tensor Products}
\label{subsec:tensor_products}

The \emph{$n$-mode product} of a tensor $\tensor{X} \in \R^{I_1 \times I_2 \times \dots \times I_N}$ and matrix $\mat{U} \in \R^{J \times I_n}$ is denoted by $\tensor{X} \times_{n} \mat{U}$ and has shape $I_1 \times \dots \times I_{n-1} \times J \times I_{n} \times \dots \times I_N$.
Elementwise, we have
\[
    \parens*{\tensor{X} \times_n \mat{U}}_{i_1 \dots i_{n-1} j i_{n+1} \dots i_N}
    =
    \sum_{i_n = 1}^{I_n} x_{i_1 i_2 \dots i_N} u_{j i_n}.
\]
In other words, each mode-$n$ fiber is multiplied by the matrix $\mat{U}$. We can express this in terms of unfolded tensors as:
\begin{equation}
\label{eqn:n_mode_product_unfolded}
    \tensor{Y} = \tensor{X} \times_{n} \mat{U}
    \iff
    \mat{Y}_{(n)} = \mat{U}\mat{X}_{(n)}.
\end{equation}

For distinct modes in a sequence of $n$-mode products, the order of the multiplication does not matter, i.e.,
\begin{align*}
    \tensor{X} \times_{m} \mat{A} \times_{n} \mat{B}
    =
    \tensor{X} \times_{n} \mat{B} \times_{m} \mat{A}
    \quad
    (m \ne n).
\end{align*}
If the modes are the same, then
\begin{align}
\label{eqn:n_mode_product_same}
    \tensor{X} \times_{n} \mat{A} \times_{n} \mat{B}
    =
    \tensor{X} \times_{n} (\mat{B} \mat{A}).
\end{align}

The \emph{inner product} of two tensors $\tensor{X}, \tensor{Y} \in \R^{I_1 \times I_2 \times \dots \times I_N}$ is the sum of the products of their entries:
\[
    \inner*{\tensor{X}, \tensor{Y}}
    =
    \sum_{i_1 = 1}^{I_1}
    \sum_{i_2 = 1}^{I_2}
    \cdots
    \sum_{i_N = 1}^{I_N}
    x_{i_1 i_2 \dots i_N}
    y_{i_1 i_2 \dots i_N}.
\]
The Frobenius norm of a tensor $\tensor{X}$ is $\norm{\tensor{X}}_{\frobenius} = \sqrt{\inner{\tensor{X}, \tensor{X}}}$.

\subsection{Tucker Decomposition}

A \emph{rank-$\mat{r}$ Tucker decomposition} approximates $\tensor{X} \in \R^{I_1 \times I_2 \times \dots \times I_N}$ by factoring it into a
\emph{core tensor} $\tensor{G} \in \R^{R_1 \times R_2 \times \dots \times R_N}$
and $N$ \emph{factor matrices} $\mat{A}^{(n)} \in \R^{I_n \times R_n}$,
where $\mat{r} = (R_1, R_2, \dots, R_N)$ is the multilinear rank.
The approximation $\widehat{\tensor{X}}$ is defined as
\[
    \widehat{\tensor{X}} = \tensor{G} \times_1 \mat{A}^{(1)} \times_2 \mat{A}^{(2)} \times_3 \dots \times_N \mat{A}^{(N)}.
\]
Finding the best rank-$\mat{r}$ Tucker decomposition is a convex and NP-hard optimization problem~\citep{hillar2013most} that aims to minimize the reconstruction error:
\begin{align}
\label{eqn:optimal_loss}
    L(\tensor{X}, \mat{r})
    \defeq
    \min_{\tensor{G}, \mat{A}^{(1)}, \dots, \mat{A}^{(N)}}
    \norm{\tensor{X} - \tensor{G} \times_1 \mat{A}^{(1)} \times_2 \mat{A}^{(2)} \times_3 \dots \times_N \mat{A}^{(N)}}_{\frobenius}^2.
\end{align}

The following result from \citet[Section 4.2]{kolda2009tensor} establishes a property that is useful for computing Tucker decompositions.
We include its proof in \Cref{app:orthonormal_columns} for completeness.

\begin{restatable}[]{lemma}{OrthonormalColumns}
\label{lem:orthonormal_columns}
Without loss of generality, we can assume each factor matrix $\mat{A}^{(n)} \in \R^{I_n \times R_n}$ has orthonormal columns.
If each $\mat{A}^{(n)} \in \R^{I_n \times R_n}$ has orthonormal columns and is fixed,
the core tensor that minimizes the reconstruction error is
\begin{equation}
\label{eqn:optimal_core_tensor}
    \tensor{G}
    = 
    \tensor{X} \times_{1} {\mat{A}^{(1)}}^\intercal \times_{2} {\mat{A}^{(2)}}^\intercal \times_{3} \dots \times_{N} {\mat{A}^{(N)}}^\intercal,
\end{equation}
and the reconstructed tensor is
\begin{align*}
  \widehat{\tensor{X}}
  &=
  \tensor{X} \times_{1} \parens{\mat{A}^{(1)} {\mat{A}^{(1)}}^\intercal}
  \times_{2} \parens{\mat{A}^{(2)} {\mat{A}^{(2)}}^\intercal}
  \times_{3}
  \dots
  \times_{N} \parens{\mat{A}^{(N)} {\mat{A}^{(N)}}^\intercal}.
\end{align*}
\end{restatable}

\subsection{Higher-Order Singular Value Decomposition (HOSVD)}

The \emph{higher-order singular value decomposition} (HOSVD) algorithm of \citet{de2000multilinear} computes each factor matrix $\mat{A}^{(n)}$ as the top-$R_n$ left singular vectors of the mode-$n$ unfolding $\mat{X}_{(n)}$. Since the factor matrices have orthonormal columns by construction, the optimal core tensor $\tensor{G}$ can be calculated directly using \eqref{eqn:optimal_core_tensor}.

\begin{algorithm}[H]
\caption{HOSVD}
\label{alg:hosvd}
\textbf{Input:} tensor $\tensor{X} \in \R^{I_1 \times I_2 \times \cdots \times I_N}$,
 rank $\mat{r} = (R_1,R_2,\dots,R_N)$

\begin{algorithmic}[1]
    \For{$n=1$ to $N$}
        \State $\mat{A}^{(n)} \gets R_n \text{ top left singular vectors of } \mat{X}_{(n)}$
    \EndFor
    \State $\tensor{G} \gets \tensor{X} \times_{1} {\mat{A}^{(1)}}^\intercal \times_{2} {\mat{A}^{(2)}}^\intercal \times_{3} \dots \times_{N} {\mat{A}^{(N)}}^\intercal$
    \State \textbf{return}
      $\tensor{G}, \mat{A}^{(1)}, \mat{A}^{(2)}, \dots, \mat{A}^{(N)}$ 
\end{algorithmic}
\end{algorithm}

A key theoretical property of HOSVD is that its reconstruction error is bounded relative to the optimal rank-$\mat{r}$ approximation. This result, shown in \citet[Property 10]{de2000multilinear} and \citet[Theorem 10.2]{hackbusch2019tensor}, is summarized in the following theorem.

\begin{theorem}[HOSVD upper bound]
\label{thm:hosvd}
For any $\tensor{X} \in \R^{I_1 \times I_2 \times \dots \times I_N}$
and $\mat{r} \in [I_1] \times \dots \times [I_N]$,
let the output of HOSVD be $\tensor{G} \in \R^{R_1 \times R_2 \times \dots \times R_N}$
and $\mat{A}^{(n)} \in \R^{I_n \times R_n}$, for $n \in [N]$.
Denoting the reconstructed rank-$\mat{r}$ tensor as
\begin{align*}
    \widehat{\tensor{X}}_{\HOSVD(\mat{r})}
    \defeq
    \tensor{G}
    \times_1 \mat{A}^{(1)}
    \times_2 \mat{A}^{(2)}
    \times_3
    \dots
    \times_N \mat{A}^{(N)},
\end{align*}
we have
\[
    \norm{\tensor{X} - \widehat{\tensor{X}}_{\HOSVD(\mat{r})}}_{\frobenius}^2
    \le
    \sum_{n=1}^N \sum_{i_n = R_n + 1}^{I_n} \parens*{\sigma_{i_n}^{(n)}}^2
    \le
    N \cdot L(\tensor{X}, \mat{r}),
\]
where $\sigma_{i_n}^{(n)}$ are the singular values of the mode-$n$ unfolding $\mat{X}_{(n)}$, sorted in nondecreasing order.
\end{theorem}

This quality guarantee makes HOSVD a popular method for initializing iterative Tucker decomposition algorithms, such as the \emph{higher-order orthogonal iteration} (HOOI) in \citet{de2000zest}.
Subsequent updates to the factor matrices and core tensor can only improve the solution, so HOSVD provides a robust starting point.
This is analogous to the $k$-means++ algorithm~\citep{arthur2006k}, which uses a provably good seeding step to find initial cluster centers before refining them with Lloyd's algorithm.

The rest of the paper is devoted to proving the tightness of the approximation result in \citet{de2000multilinear}.
We show that it is impossible to improve this guarantee to $N / (1+\varepsilon)$, for any $\varepsilon > 0$.
To the best of our knowledge, this fundamental result about the optimality of the HOSVD analysis has been missing from the literature, despite the algorithm's prominence.

\section{Lower Bounding the Reconstruction Error}
\label{sec:lower_bound}

\subsection{Simple Construction}
\label{subsec:construction}

We start by constructing an adversarial symmetric input to demonstrate a failure case for the HOSVD algorithm.

\begin{tcolorbox}
\begin{definition}[Simple construction]
\label{def:construction}
For any $N \ge 2$ and $\varepsilon > 0$,
let the symmetric tensor $\tensor{X} \in \R^{3 \times 3 \times \dots \times 3}$
be the sum of two components, $\tensor{X} = \tensor{Y} + \tensor{Z}$,
defined as follows:
\begin{itemize}
    \item The \emph{top component} $\tensor{Y}$ has one nonzero value:
    \[
        y[i_1,i_2,\dots,i_N]
        =
        \begin{cases}
            \sqrt{1 + \varepsilon} & \text{if $(i_1, i_2, \dots, i_N) = (1,1,\dots,1)$,} \\
            0 & \text{otherwise.}
        \end{cases}
    \]
    \item The \emph{bottom component} $\tensor{Z}$ has $N$ nonzeros values:
    \[
        z[i_1,i_2,\dots,i_N]
        =
        \begin{cases}
            1 & \text{if $(i_1,i_2,\dots, i_N) = (2,3,\dots,3)$,} \\
            1 & \text{if $(i_1,i_2,\dots, i_N) = (3,2,\dots,3)$,} \\
            \vdots \\
            1 & \text{if $(i_1,i_2,\dots, i_N) = (3,3,\dots,2)$,} \\
            0 & \text{otherwise.}
        \end{cases}
  \]
\end{itemize}
We use tensor $\tensor{X}$ with the target rank $\mat{r} = (2,2,\dots,2)$ to demonstrate the worst-case behavior of HOSVD.
\end{definition}
\end{tcolorbox}

\paragraph{Examples.}
For $N=2$, we have
\begin{align*}
  \mat{X}
  =
  \begin{bmatrix}
    \sqrt{1 + \varepsilon} & 0 & 0 \\
    0 & 0 & 1 \\
    0 & 1 & 0
  \end{bmatrix}.
\end{align*}
For $N=3$, the slices of $\tensor{X}$ are
\begin{align*}
  \mat{X}_{1::}
  &=
  \begin{bmatrix}
    \sqrt{1 + \varepsilon} & 0 & 0 \\
    0 & 0 & 0 \\
    0 & 0 & 0
  \end{bmatrix}
  \quad
  \mat{X}_{2::}
  =
  \begin{bmatrix}
    0 & 0 & 0 \\
    0 & 0 & 0 \\
    0 & 0 & 1
  \end{bmatrix}
  \quad
  \mat{X}_{3::}
  =
  \begin{bmatrix}
    0 & 0 & 0 \\
    0 & 0 & 1 \\
    0 & 1 & 0
  \end{bmatrix}.
\end{align*}

\paragraph{Key idea.}
Our construction exploits the greedy nature of HOSVD.
The main idea is to start with tensor $\tensor{Z}$ since it is perfectly recoverable by a rank-$(2,2,\dots,2)$ Tucker decomposition (\Cref{lem:alternative_decomposition}).
Then we modify $\tensor{Z}$ by increasing the value at $(1,1,\dots,1)$
from $0 \rightarrow \sqrt{1 + \varepsilon}$, giving us $\tensor{X}$.
HOSVD greedily focuses on this new high-value entry, i.e., the top component $\tensor{Y}$.
Due to the rank constraint $\mat{r} = (2,2,\dots,2)$,
recovering the top component forces HOSVD to drop the entire bottom component (\Cref{lem:hosvd_error}).
Since HOSVD computes each factor matrix $\mat{A}^{(n)}$ \emph{independently}
using the SVD of $\mat{X}_{(n)}$,
our symmetric construction tricks it into making the same mistake for all $N$ modes.
These greedy errors accumulate, leading to a reconstruction error of $N$
in the squared Frobenius norm, while $(1 + \varepsilon)$ is achievable.

\subsection{HOSVD Analysis}

Here we show that $\widehat{\tensor{X}}_{\HOSVD(\mat{r})} = \tensor{Y}$,
i.e., running HOSVD on our adversarial input only recovers the top component.

\begin{fact}
\label{fact:symmetry}
A symmetric tensor $\tensor{X}$ has identical mode-$n$ unfoldings:
\[
    \mat{X}_{(1)} = \mat{X}_{(2)} = \dots = \mat{X}_{(N)}.
\]
Consequently, if their left singular matrices $\mat{U}^{(n)}$ are computed using a deterministic SVD algorithm in HOSVD, they are also identical:
\[
  \mat{U}^{(1)} = \mat{U}^{(2)} = \dots = \mat{U}^{(N)}.
\]
\end{fact}

A deterministic SVD is critical to our analysis because it \emph{consistently breaks ties} caused by repeated singular values when computing the factor matrix $\mat{A}^{(n)}$ in each step of HOSVD.
We explore this in more detail in \Cref{app:svd} when exploring the difference between SVD and HOSVD for $N = 2$.

\begin{lemma}
\label{lem:left_singular_values}
For any $N \ge 2$ and $n \in [N]$,
the matrix $\mat{X}_{(n)}$ has the following left singular vectors and singular values:
\begin{itemize}
  \item $(\mat{e}_1, \sqrt{1+\varepsilon})$,
  \item $(\mat{e}_2, 1)$,
  \item $(\mat{e}_3, \sqrt{N-1})$.
\end{itemize}
\end{lemma}

\begin{proof}
By symmetry (\Cref{fact:symmetry}), it suffices to analyze $\mat{X}_{(1)}$.
Each column of $\mat{X}_{(1)}$ has at most one nonzero entry, so the matrix $\mat{X}_{(1)} \mat{X}_{(1)}^\intercal$ is diagonal:
\begin{align*}
    \mat{X}_{(1)} \mat{X}_{(1)}^\intercal
    =
    \begin{bmatrix}
        1 + \varepsilon & 0 & 0 \\
        0 & 1 & 0 \\
        0 & 0 & N - 1
    \end{bmatrix}.
\end{align*}
The eigendecomposition of a diagonal matrix has the
standard basis vectors $\mat{e}_i$ as its eigenvectors, so
\begin{align*}
    \mat{X}_{(1)} \mat{X}_{(1)}^\intercal
    =
    \mat{U}^{(1)}
    \mat{\Sigma}^2
    {\mat{U}^{(1)}}^\intercal
    =
    \begin{bmatrix}
      1 & 0 & 0 \\
      0 & 1 & 0 \\
      0 & 0 & 1 \\
    \end{bmatrix}
    \begin{bmatrix}
        1 + \varepsilon & 0 & 0 \\
        0 & 1 & 0 \\
        0 & 0 & N - 1
    \end{bmatrix}
    \begin{bmatrix}
      1 & 0 & 0 \\
      0 & 1 & 0 \\
      0 & 0 & 1 \\
    \end{bmatrix}^\intercal,
\end{align*}
which proves the claim.
\end{proof}

We now analyze HOSVD for our adversarial input $\tensor{X} = \tensor{Y} + \tensor{Z}$ with multilinear rank $\mat{r} = (2,2,\dots,2)$.

\begin{lemma}
\label{lem:hosvd_error}
For any $N \ge 2$ and $\varepsilon > 0$,
the HOSVD of our construction $(\tensor{X}, \mat{r})$ only recovers the top component, i.e.,
\[
  \widehat{\tensor{X}}_{\HOSVD(\mat{r})}
  =
  \tensor{Y}.
\]
The reconstruction error is
\[
  \norm{\tensor{X} - \widehat{\tensor{X}}_{\HOSVD(\mat{r})}}_{\frobenius}^2 = N.
\]
\end{lemma}

\begin{proof}
Using the formula for the optimal reconstructed tensor from fixed orthonormal factor matrices in \Cref{lem:orthonormal_columns},
\begin{align*}
  \widehat{\tensor{X}}_{\HOSVD(\mat{r})}
  &=
  \tensor{X} \times_{1} \parens{\mat{A}^{(1)} {\mat{A}^{(1)}}^\intercal}
  \times_{2} \parens{\mat{A}^{(2)} {\mat{A}^{(2)}}^\intercal}
  \times_{3}
  \dots
  \times_{N} \parens{\mat{A}^{(N)} {\mat{A}^{(N)}}^\intercal} \\
  &=
  \tensor{X} \times_{1} \mat{M}
  \times_{2} \mat{M}
  \times_{3}
  \dots
  \times_{N} \mat{M},
\end{align*}
where $\mat{M} = \mat{A}^{(1)} {\mat{A}^{(1)}}^\intercal$
since all of the factor matrices are equal by \Cref{fact:symmetry}.

We show that there are two possible values for $\mat{M}$ given our choice of $\mat{r} = (2,2,\dots,2)$:
\begin{itemize}

\item \textbf{Case 1 $(N=2)$.}
By \Cref{lem:left_singular_values},
the singular values of $\mat{X}_{(n)}$ are $(\sqrt{1 + \varepsilon}, 1, 1)$.
Since $\sigma_2 = \sigma_3$, we have
\begin{align*}
  \mat{A}^{(n)}
  \in
  \set*{
  \begin{bmatrix}
    1 & 0 \\
    0 & 1 \\
    0 & 0 \\
  \end{bmatrix},
  \begin{bmatrix}
    1 & 0 \\
    0 & 0 \\
    0 & 1 \\
  \end{bmatrix}
  }
  \implies
  \mat{M}
  \in
  \set*{
  \begin{bmatrix}
    1 & 0 & 0\\
    0 & 1 & 0\\
    0 & 0 & 0\\
  \end{bmatrix},
  \begin{bmatrix}
    1 & 0 & 0\\
    0 & 0 & 0\\
    0 & 0 & 1\\
  \end{bmatrix}
  }.
\end{align*}

\item \textbf{Case 2 $(N \ge 3)$.}
By \Cref{lem:left_singular_values},
the singular values of $\mat{X}_{(n)}$ are $(\sqrt{N-1}, \sqrt{1 + \varepsilon}, 1)$, so we have
\begin{align*}
  \mat{A}^{(n)}
  \in
  \set*{
  \begin{bmatrix}
    0 & 1 \\
    0 & 0 \\
    1 & 0 \\
  \end{bmatrix},
  \begin{bmatrix}
    1 & 0 \\
    0 & 0 \\
    0 & 1 \\
  \end{bmatrix}
  }
  \implies
  \mat{M}
  =
  \begin{bmatrix}
    1 & 0 & 0\\
    0 & 0 & 0\\
    0 & 0 & 1\\
  \end{bmatrix},
\end{align*}
which is covered by the $N=2$ case.
\end{itemize}

Next we use the unfolded interpretation of the $n$-mode product in \eqref{eqn:n_mode_product_unfolded} to prove that only the top component~$\tensor{Y}$ survives HOSVD reconstruction.
For any $n \in [N]$,
\begin{align*}
    \tensor{T} = \tensor{X} \times_{n} \mat{M}
    \iff
    \mat{T}_{(n)} = \mat{M} \mat{X}_{(n)}.
\end{align*}
Since $\mat{M} = \mat{I}_3 - \mat{e}_{i} \mat{e}_{i}^\intercal$ for $i=2$ or $i=3$,
the $n$-mode product with $\mat{M}$ eliminates the $i$-th row of $\mat{X}_{(n)}$.
Recalling that
\[
    \widehat{\tensor{X}}_{\HOSVD(\mat{r})}
    =
    \tensor{X} \times_{1} \mat{M} \times_{2} \mat{M} \times_{3} \dots \times_{N} \mat{M},
\]
the bottom component gets zeroed out since the index of each of its nonzeros
contains a $2$ and $3$.
Observe that we get the same result for either choice of $\mat{M}$.
In contrast, the top component $\tensor{Y}$ survives since we always have $m_{11} = 1$
and its nonzero value is at index $(1,1,\dots,1)$.

Putting everything together, we have
\begin{align*}
  \widehat{\tensor{X}}_{\HOSVD(\mat{r})} 
  =
  \tensor{Y}
  \implies
  \norm{\tensor{X} - \widehat{\tensor{X}}_{\HOSVD(\mat{r})}}_{\frobenius}^2
  =
  \norm{(\tensor{Y} + \tensor{Z}) - \tensor{Y}}_{\frobenius}^2
  =
  \norm{\tensor{Z}}_{\frobenius}^2
  =
  N,
\end{align*}
as desired.
\end{proof}

\subsection{Alternative Decomposition}

We now show that the bottom component can be \emph{perfectly reconstructed}
with a rank-$(2,2,\dots,2)$ Tucker decomposition.
Interestingly, this is the decomposition we obtain by applying HOSVD directly to $\tensor{Z}$,
i.e., $\widehat{\tensor{Z}}_{\HOSVD(\mat{r})} = \tensor{Z}$.

\begin{lemma}
\label{lem:alternative_decomposition}
For any $N \ge 2$, there is a rank-$(2,2,\dots,2)$ Tucker decomposition
that perfectly reconstructs the bottom component $\tensor{Z}$.
\end{lemma}

\begin{proof}
For each $n \in [N]$, let the factor matrix be
\begin{equation}
\label{eqn:alternative_factor_matrix}
  \mat{A}^{(n)}
  =
  \begin{bmatrix}
    0 & 0 \\
    0 & 1 \\
    1 & 0 \\
  \end{bmatrix}.
\end{equation}
The optimal reconstructed tensor from fixed orthonormal factor matrices is
\begin{align*}
  \widehat{\tensor{Z}}_{\HOSVD(\mat{r})}
  &=
  \tensor{Z} \times_{1} \parens{\mat{A}^{(1)} {\mat{A}^{(1)}}^\intercal}
  \times_{2} \parens{\mat{A}^{(2)} {\mat{A}^{(2)}}^\intercal}
  \times_{3}
  \dots
  \times_{N} \parens{\mat{A}^{(N)} {\mat{A}^{(N)}}^\intercal} \\
  &=
  \tensor{Z} \times_{1} \mat{M}
  \times_{2} \mat{M}
  \times_{3}
  \dots
  \times_{N} \mat{M},
\end{align*}
by \Cref{lem:orthonormal_columns}, where
\[
  \mat{M}
  =
  \begin{bmatrix}
  0 & 0 & 0 \\
  0 & 1 & 0 \\
  0 & 0 & 1 
  \end{bmatrix}.
\]
Using the unfolding interpretation of the $n$-mode product
(like in the proof of \Cref{lem:hosvd_error}),
we have $\mat{M} = \mat{I}_{3} - \mat{e}_{1} \mat{e}_{1}^\intercal$,
so the only entries of $\tensor{Z}$ that survive reconstruction have indices in $\{2,3\}^N$,
which is precisely the bottom component.
\end{proof}

\begin{remark}
The factor matrices in \eqref{eqn:alternative_factor_matrix} contain the basis vectors $\{\mat{e}_2, \mat{e}_3\}$.
In contrast, if we run HOSVD on $(\tensor{X}, \mat{r})$,
it greedily (and suboptimally) chooses bases 
$\{\mat{e}_1, \mat{e}_2\}$ or $\{\mat{e}_{1}, \mat{e}_3\}$
because of the top component's singular value $\sqrt{1 + \varepsilon}$.
\end{remark}

\subsection{Proof of \Cref{thm:hosvd_lower_bound}}

We are ready to lower bound the approximation ratio of HOSVD.
Since our result holds for any $\varepsilon > 0$, this proves that the
$N$-approximation of \citet{de2000multilinear} is tight and cannot be improved.

\HosvdLowerBound*

\begin{proof}
This is a direct consequence of \Cref{lem:hosvd_error} and \Cref{lem:alternative_decomposition}.
For $\tensor{X} = \tensor{Y} + \tensor{Z}$ and $\mat{r} = (2,2,\dots,2)$,
we have
\begin{align*}
    \norm{\tensor{X} - \widehat{\tensor{X}}_{\HOSVD(\mat{r})}}_{\frobenius}^2
    =
    N
    \ge
    \frac{N}{1 + \varepsilon} \cdot L(\tensor{X}, \mat{r}),
\end{align*}
which proves the claim.
\end{proof}

\section{Extensions to Adaptive Methods}
\label{sec:extensions}

We now consider two \emph{adaptive} methods for Tucker decomposition that, in practice, improve on HOSVD.
We show that the worst-case approximation ratio of these methods is also $N/(1+\epsilon)$,
for any $\varepsilon > 0$.
To do this, we construct a more intricate adversarial tensor.

The methods we consider are \emph{sequentially truncated HOSVD} (ST-HOSVD) of \citet{vannieuwenhoven2012new}
the and \emph{higher-order orthogonal iteration} (HOOI) of \citet{de2000zest}.
HOOI is version of the alternating least squares (ALS) algorithm that forces the factor matrices be orthogonal. Interestingly, the ALS updates naturally produce orthogonal matrices.
Therefore, our construction shows that starting from an initial Tucker decomposition given by HOSVD or ST-HOSVD, ALS and HOOI do not improve the error.

\subsection{Advanced Construction}
\label{subsec:construction-2}

We start by presenting a more advanced tensor construction.

\begin{tcolorbox}
\begin{definition}[Advanced construction]
\label{def:construction-sthosvd}
For any $N \ge 3$ and $\varepsilon > 0$,
let the symmetric tensor $\tensor{X} \in \R^{4 \times 4 \times \dots \times 4}$
be the sum of three components, $\tensor{X} = \tensor{Y} + \tensor{Z} + \tensor{L}$,
defined as follows:
\begin{itemize}
    \item The \emph{top component} $\tensor{Y}$ has one nonzero value:
    \[
        y[i_1,i_2,\dots,i_N]
        =
        \begin{cases}
            \sqrt{1 + \varepsilon} & \text{if $(i_1, i_2, \dots, i_N) = (1,1,\dots,1)$,} \\
            0 & \text{otherwise.}
        \end{cases}
    \]
    \item The \emph{bottom component} $\tensor{Z}$ has $N$ nonzeros values:
    \[
        z[i_1,i_2,\dots,i_N]
        =
        \begin{cases}
            1 & \text{if $(i_1,i_2,\dots, i_N) = (4,3,\dots,3)$,} \\
            1 & \text{if $(i_1,i_2,\dots, i_N) = (3,4,\dots,3)$,} \\
            \vdots \\
            1 & \text{if $(i_1,i_2,\dots, i_N) = (3,3,\dots,4)$,} \\
            0 & \text{otherwise.}
        \end{cases}
  \]
  \item The \emph{middle component} $\tensor{L}$ has $N$ nonzeros values:
    \[
        \ell[i_1,i_2,\dots,i_N]
        =
        \begin{cases}
            \sqrt{1 + \varepsilon} & \text{if $(i_1,i_2,\dots, i_N) = (3,2,\dots,2)$,} \\
            \sqrt{1 + \varepsilon} & \text{if $(i_1,i_2,\dots, i_N) = (2,3,\dots,2)$,} \\
            \vdots \\
            \sqrt{1 + \varepsilon} & \text{if $(i_1,i_2,\dots, i_N) = (2,2,\dots,3)$,} \\
            0 & \text{otherwise.}
        \end{cases}
  \]
\end{itemize}
We use tensor $\tensor{X}$ with the target rank $\mat{r} = (3,3,\dots,3)$ to demonstrate the worst-case behavior of ST-HOSVD.
\end{definition}
\end{tcolorbox}

\paragraph{Example.}
For $N=3$, the slices of $\tensor{X}$ are
\begin{align*}
  \mat{X}_{1::}
  &=
  \begin{bmatrix}
    \sqrt{1 + \varepsilon} & 0 & 0 & 0 \\
    0 & 0 & 0 & 0 \\
    0 & 0 & 0 & 0 \\
    0 & 0 & 0 & 0
  \end{bmatrix}
  \quad
  \mat{X}_{2::}
  =
  \begin{bmatrix}
    0 & 0 & 0 & 0 \\
    0 & 0 & \sqrt{1 + \varepsilon} & 0 \\
    0 & \sqrt{1 + \varepsilon} & 0 & 0 \\
    0 & 0 & 0 & 0
  \end{bmatrix}
  \\
  \mat{X}_{3::}
  & =
  \begin{bmatrix}
    0 & 0 & 0 & 0\\
    0 & \sqrt{1 + \varepsilon} & 0 & 0 \\
    0 & 0 & 0 & 1 \\
    0 & 0 & 1 & 0
  \end{bmatrix}
  \quad
  \mat{X}_{4::}
  =
  \begin{bmatrix}
    0 & 0 & 0 & 0\\
    0 & 0 & 0 & 0 \\
    0 & 0 & 1 & 0 \\
    0 & 0 & 0 & 0
  \end{bmatrix}.
\end{align*}

\subsection{ST-HOSVD}
\label{subsec:st-hosvd}

The main difference between HOSVD and ST-HOSVD is that after computing the factor matrix for each mode $n$, we project the tensor onto that factor matrix (see \Cref{alg:st-hosvd}).
Consequently, unlike HOSVD, the decomposition computed by ST-HOSVD depends on the order in which the modes are processed. For symmetric tensors, however, such as our construction in \Cref{def:construction-sthosvd},
permuting the order of the modes only permutes the factor matrices,
which does not affect the reconstruction error.
Therefore, we restrict attention to the order $1,2,\dots,N$ for the modes.

\begin{algorithm}[H]
\caption{ST-HOSVD}
\label{alg:st-hosvd}
\textbf{Input:} tensor $\tensor{X} \in \R^{I_1 \times I_2 \times \cdots \times I_N}$,
 rank $\mat{r} = (R_1,R_2,\dots,R_N)$

\begin{algorithmic}[1]
    \State $\tensor{X}^{(1)} \gets \tensor{X}$
    \For{$n=1$ to $N$}
        \State $\mat{A}^{(n)} \gets R_n \text{ top left singular vectors of } \mat{X}^{(n)}_{(n)}$
        \State $\tensor{X}^{(n+1)} \gets \tensor{X}^{(1)} \times_{n} {\mat{A}^{(n)}}^\intercal$
    \EndFor
    \State $\tensor{G} \gets \tensor{X} \times_{1} {\mat{A}^{(1)}}^\intercal \times_{2} {\mat{A}^{(2)}}^\intercal \times_{3} \dots \times_{N} {\mat{A}^{(N)}}^\intercal$
    \State \textbf{return}
      $\tensor{G}, \mat{A}^{(1)}, \mat{A}^{(2)}, \dots, \mat{A}^{(N)}$ 
\end{algorithmic}
\end{algorithm}

Proving a lower bound for ST-HOSVD is more challenging because the tensor $\tensor{X}^{(n)}$
changes in each iteration. This necessitates the more intricate construction in \Cref{def:construction-sthosvd}. Nevertheless, the overall strategy remains the same: we can trick the algorithm into selecting the wrong singular vectors due to its greedy nature.

\SthosvdLowerBound*
\begin{proof}
We use the tensor constructed in \Cref{def:construction-sthosvd} with $\mat{r} = (3,3,\dots,3)$.
For the original tensor, we have
\begin{align*}
    \mat{X}^{(1)}_{(1)} {\mat{X}^{(1)}_{(1)}}^\intercal
    =
    \begin{bmatrix}
        1 + \varepsilon & 0 & 0 & 0 \\
        0 & (N-1)(1+\varepsilon) & 0 & 0 \\
        0 & 0 & N +\varepsilon & 0 \\
        0 & 0 & 0 & 1
    \end{bmatrix}.
\end{align*}
Therefore,
the matrix $\mat{X}^{(1)}_{(1)}$ has the following left singular vectors and singular values:
\begin{itemize}
  \item $(\mat{e}_1, \sqrt{1+\varepsilon})$,
  \item $(\mat{e}_2, \sqrt{(N-1)(1+\varepsilon)})$,
  \item $(\mat{e}_3, \sqrt{N +\varepsilon})$,
  \item $(\mat{e}_4, 1)$.
\end{itemize}
Therefore, ST-HOSVD selects $\mat{A}^{(1)} = \begin{bmatrix}
    \mat{e}_1 & \mat{e}_2 & \mat{e}_3
\end{bmatrix}$. Then $\tensor{X}^{(2)} \in \R^{3 \times 4 \times 4 \times \cdots \times 4}$ is the tensor obtained from $\tensor{X}^{(1)}$ by removing the slice $(4,:,:,\ldots,:)$,
which includes only one nonzero entry.
One can then observe that
\begin{align*}
    \mat{X}_{(2)}^{(2)} (\mat{X}_{(2)}^{(2)})^\intercal
    =
    \begin{bmatrix}
        1 + \varepsilon & 0 & 0 & 0 \\
        0 & (N-1)(1+\varepsilon) & 0 & 0 \\
        0 & 0 & N - 1 + \varepsilon & 0 \\
        0 & 0 & 0 & 1
    \end{bmatrix}.
\end{align*}
Therefore, ST-HOSVD also selects $\mat{A}^{(2)} = \begin{bmatrix}
    \mat{e}_1 & \mat{e}_2 & \mat{e}_3
\end{bmatrix}$ and consequently,
\begin{align*}
    \mat{X}_{(3)}^{(3)} (\mat{X}_{(3)}^{(3)})^\intercal
    =
    \begin{bmatrix}
        1 + \varepsilon & 0 & 0 & 0 \\
        0 & (N-1)(1+\varepsilon) & 0 & 0 \\
        0 & 0 & N - 2 +\varepsilon & 0 \\
        0 & 0 & 0 & 1
    \end{bmatrix}.
\end{align*}
Continuing this inductively, we arrive at
\begin{align}
\label{eq:xn-gram}
    \mat{X}_{(N)}^{(N)} (\mat{X}_{(N)}^{(N)})^\intercal
    =
    \begin{bmatrix}
        1 + \varepsilon & 0 & 0 & 0 \\
        0 & (N-1)(1+\varepsilon) & 0 & 0 \\
        0 & 0 & 1+\varepsilon & 0 \\
        0 & 0 & 0 & 1
    \end{bmatrix}.
\end{align}
The order that we select the dimensions does not matter due to symmetry of the example and the above argument.
Therefore, ST-HOSVD chooses
$\begin{bmatrix}
    \mat{e}_1 & \mat{e}_2 & \mat{e}_3
\end{bmatrix}$ for all factor matrices, which is equivalent to selecting the subtensor corresponding to $(1:3,1:3,\cdots, 1:3)$.
It follows that
\[
    \tensor{X} - \widehat{\tensor{X}}_{\STHOSVD(\mat{r})} = \tensor{Z},
\]
where $\tensor{Z}$ is as defined in \Cref{def:construction-sthosvd}.
Therefore,
\[
\norm{\tensor{X} - \widehat{\tensor{X}}_{\STHOSVD(\mat{r})}}_{\frobenius}^2 = N.
\]
However, taking $\tensor{X}^*$ to be the low-rank tensor corresponding to all factor matrices equal to
$\begin{bmatrix}
    \mat{e}_2 & \mat{e}_3 & \mat{e}_4
\end{bmatrix}$,
which is equivalent to selecting the subtensor corresponding to $(2:4,2:4,\cdots, 2:4)$,
gives the following reconstruction error.
\[
    \norm{\tensor{X} - \tensor{X}^*}_{\frobenius}^2 = 1+\epsilon.
\]
The proof of the claim follows.
\end{proof}

\subsection{HOOI}
\label{subsec:hooi}

The HOOI algorithm is first introduced in \citet{de2000zest}.
However, we use a simpler description due to \citet{kolda2009tensor} presented in \Cref{alg:hooi}. The HOOI algorithm is an ALS algorithm that ensures the factor matrices have orthonormal columns.
Denoting $\tensor{G} = \tensor{X} \times_1 \mat{A}^{(1)} \times_2 \cdots \times_N \mat{A}^{(N)}$, for the current values of the factor matrices,
HOOI is equivalent to solving the following linear regression problem in iteration $n$ of the inner loop:
\[
\min_{\substack{\mat{A}^{(n)} \in \R^{I_n \times R_n}:\\ \mat{A}^{(n)} \text{ has orthonormal columns}}} \norm{\tensor{X} - \tensor{G} \times_1 \mat{A}^{(1)} \times_2 \mat{A}^{(2)} \times_3 \dots \times_N \mat{A}^{(N)}}_F^2,
\]
which without the orthogonality condition is a ordinary linear regression problem.

\begin{algorithm}[H]
\caption{HOOI}
\label{alg:hooi}
\textbf{Input:} tensor $\tensor{X} \in \R^{I_1 \times I_2 \times \cdots \times I_N}$,
 rank $\mat{r} = (R_1,R_2,\dots,R_N)$

\begin{algorithmic}[1]
    \State Initialize $\mat{A}^{(n)} \in \R^{I_n \times R_n}$ for $n \in [N]$ using HOSVD
    \Repeat
    \For{$n=1$ to $N$}
        \State $\tensor{B} \gets \tensor{X} \times_{1} \mat{A}^{(1)}\times_{2} \cdots \times_{n-1} \mat{A}^{(n-1)} \times_{n+1} \mat{A}^{(n+1)} \times_{n+2} \cdots \times_{N} \mat{A}^{(N)}$
        \State $\mat{A}^{(n)} \gets R_n \text{ top left singular vectors of } \mat{B}_{(n)}$
    \EndFor
    \Until{converged or iteration limit}
    \State $\tensor{G} \gets \tensor{X} \times_{1} {\mat{A}^{(1)}}^\intercal \times_{2} {\mat{A}^{(2)}}^\intercal \times_{3} \dots \times_{N} {\mat{A}^{(N)}}^\intercal$
    \State \textbf{return}
      $\tensor{G}, \mat{A}^{(1)}, \mat{A}^{(2)}, \dots, \mat{A}^{(N)}$ 
\end{algorithmic}
\end{algorithm}

\HooiLowerBound*
\begin{proof}
Since the HOOI algorithm initializes the factor matrices with HOSVD, initially
\[
\mat{A}^{(1)} = \cdots = \mat{A}^{(N)} = \begin{bmatrix}
    \mat{e}_1 & \mat{e}_2 & \mat{e}_3
\end{bmatrix}.
\]
We show by induction that the factor matrices do not change over the iterations of HOOI. Consider $n\in[N]$. By the induction hypothesis, all factor matrices are equal to
$\begin{bmatrix}
    \mat{e}_1 & \mat{e}_2 & \mat{e}_3
\end{bmatrix}$.
Therefore, for 
\[
\tensor{B} = \tensor{X} \times_{1} \mat{A}^{(1)}\times_{2} \cdots \times_{n-1} \mat{A}^{(n-1)} \times_{n+1} \mat{A}^{(n+1)} \times_{n+2} \cdots \times_{N} \mat{A}^{(N)},
\]
we have that $\mat{B}_{(n)}$ is equal to $\mat{X}^{(N)}_{(N)}$ that is produced in the ST-HOSVD algorithm.
Therefore, by \eqref{eq:xn-gram},
\[
\mat{B}_{(n)}{\mat{B}_{(n)}}^\intercal = \mat{X}_{(N)}^{(N)} (\mat{X}_{(N)}^{(N)})^\intercal
    =
    \begin{bmatrix}
        1 + \varepsilon & 0 & 0 & 0 \\
        0 & (N-1)(1+\varepsilon) & 0 & 0 \\
        0 & 0 & 1+\varepsilon & 0 \\
        0 & 0 & 0 & 1
    \end{bmatrix}.
\]
It follows that the updated $\mat{A}^{(n)}$ is also equal to $\begin{bmatrix}
    \mat{e}_1 & \mat{e}_2 & \mat{e}_3
\end{bmatrix}$, which concludes the proof.
\end{proof}

\section{Conclusion}

We prove that the classic approximation guarantee for the HOSVD algorithm in \citet{de2000multilinear} is tight.
To do this, we construct a tensor that exploits the greedy nature of HOSVD and forces it to make independent myopic choices in each mode,
resulting in a provable lower bound for the approximation factor of $N/(1+\varepsilon)$, for any $\varepsilon>0$.
We then extend this idea with a more intricate construction to show that the approximation ratio for ST-HOSVD in \citet{vannieuwenhoven2012new}
and HOOI in \citet{de2000zest} are at least $N / (1 + \varepsilon)$, for any $\varepsilon > 0$.
To the best of our knowledge, these are the first results showing that the approximation guarantees of HOSVD, ST-HOSVD, and HOOI match their upper bounds and cannot be improved.

We conclude with an open problem.
As we have demonstrated,
existing low-rank Tucker decomposition algorithms can output solutions with an approximation ratio as poor as $N$.
Thus, a natural open problem is to
give a polynomial-time algorithm (in the size of the tensor $I_1 I_2 \cdots I_N)$
with an approximation ratio strictly less than $N$,
or show that the approximation ratio cannot be improved in polynomial time under standard complexity-theoretic assumptions.

\bibliographystyle{plainnat}
\bibliography{references}

\begin{thebibliography}{18}
\providecommand{\natexlab}[1]{#1}
\providecommand{\url}[1]{\texttt{#1}}
\expandafter\ifx\csname urlstyle\endcsname\relax
  \providecommand{\doi}[1]{doi: #1}\else
  \providecommand{\doi}{doi: \begingroup \urlstyle{rm}\Url}\fi

\bibitem[Arthur and Vassilvitskii(2006)]{arthur2006k}
David Arthur and Sergei Vassilvitskii.
\newblock \texttt{k-means++:} {T}he advantages of careful seeding.
\newblock Technical report, Stanford, 2006.

\bibitem[Bader et~al.(2023)Bader, Kolda, et~al.]{matlab}
Brett~W. Bader, Tamara~G. Kolda, et~al.
\newblock {MATLAB Tensor Toolbox Version 3.6}.
\newblock Available online, September 2023.
\newblock URL \url{https://www.tensortoolbox.org}.

\bibitem[De~Lathauwer et~al.(2000{\natexlab{a}})De~Lathauwer, De~Moor, and Vandewalle]{de2000multilinear}
Lieven De~Lathauwer, Bart De~Moor, and Joos Vandewalle.
\newblock A multilinear singular value decomposition.
\newblock \emph{SIAM Journal on Matrix Analysis and Applications}, 21\penalty0 (4):\penalty0 1253--1278, 2000{\natexlab{a}}.

\bibitem[De~Lathauwer et~al.(2000{\natexlab{b}})De~Lathauwer, De~Moor, and Vandewalle]{de2000zest}
Lieven De~Lathauwer, Bart De~Moor, and Joos Vandewalle.
\newblock On the best rank-1 and rank-$(r_1, r_2,..., r_n)$ approximation of higher-order tensors.
\newblock \emph{SIAM Journal on Matrix Analysis and Applications}, 21\penalty0 (4):\penalty0 1324--1342, 2000{\natexlab{b}}.

\bibitem[Eld{\'e}n and Savas(2009)]{elden2009newton}
Lars Eld{\'e}n and Berkant Savas.
\newblock A {N}ewton--{G}rassmann method for computing the best multilinear rank-$(r_1, r_2, r_3)$ approximation of a tensor.
\newblock \emph{SIAM Journal on Matrix Analysis and Applications}, 31\penalty0 (2):\penalty0 248--271, 2009.

\bibitem[Eld{\'e}n and Savas(2011)]{elden2011perturbation}
Lars Eld{\'e}n and Berkant Savas.
\newblock Perturbation theory and optimality conditions for the best multilinear rank approximation of a tensor.
\newblock \emph{SIAM Journal on Matrix Analysis and Applications}, 32\penalty0 (4):\penalty0 1422--1450, 2011.

\bibitem[Fahrbach et~al.(2022)Fahrbach, Fu, and Ghadiri]{fahrbach2022subquadratic}
Matthew Fahrbach, Gang Fu, and Mehrdad Ghadiri.
\newblock Subquadratic {K}ronecker regression with applications to tensor decomposition.
\newblock \emph{Advances in Neural Information Processing Systems}, 35:\penalty0 28776--28789, 2022.

\bibitem[Frandsen and Ge(2022)]{frandsen2022optimization}
Abraham Frandsen and Rong Ge.
\newblock Optimization landscape of tucker decomposition.
\newblock \emph{Mathematical Programming}, 193\penalty0 (2):\penalty0 687--712, 2022.

\bibitem[Ghadiri et~al.(2023)Ghadiri, Fahrbach, Fu, and Mirrokni]{ghadiri2023approximately}
Mehrdad Ghadiri, Matthew Fahrbach, Gang Fu, and Vahab Mirrokni.
\newblock Approximately optimal core shapes for tensor decompositions.
\newblock In \emph{International Conference on Machine Learning}, pages 11237--11254. PMLR, 2023.

\bibitem[Hackbusch(2019)]{hackbusch2019tensor}
Wolfgang Hackbusch.
\newblock \emph{Tensor Spaces and Numerical Tensor Calculus}, volume~56.
\newblock Springer, 2nd edition, 2019.

\bibitem[Hillar and Lim(2013)]{hillar2013most}
Christopher~J Hillar and Lek-Heng Lim.
\newblock Most tensor problems are {NP}-hard.
\newblock \emph{Journal of the ACM}, 60\penalty0 (6):\penalty0 1--39, 2013.

\bibitem[Jang and Kang(2021)]{jang2021fast}
Jun-Gi Jang and U~Kang.
\newblock Fast and memory-efficient {T}ucker decomposition for answering diverse time range queries.
\newblock In \emph{Proceedings of the 27th ACM SIGKDD Conference on Knowledge Discovery \& Data Mining}, pages 725--735, 2021.

\bibitem[Kolda and Bader(2009)]{kolda2009tensor}
Tamara~G Kolda and Brett~W Bader.
\newblock Tensor decompositions and applications.
\newblock \emph{SIAM Review}, 51\penalty0 (3):\penalty0 455--500, 2009.

\bibitem[Rabanser et~al.(2017)Rabanser, Shchur, and G{\"u}nnemann]{rabanser2017introduction}
Stephan Rabanser, Oleksandr Shchur, and Stephan G{\"u}nnemann.
\newblock Introduction to tensor decompositions and their applications in machine learning.
\newblock \emph{arXiv preprint arXiv:1711.10781}, 2017.

\bibitem[Sidiropoulos et~al.(2017)Sidiropoulos, De~Lathauwer, Fu, Huang, Papalexakis, and Faloutsos]{sidiropoulos2017tensor}
Nicholas~D Sidiropoulos, Lieven De~Lathauwer, Xiao Fu, Kejun Huang, Evangelos~E Papalexakis, and Christos Faloutsos.
\newblock Tensor decomposition for signal processing and machine learning.
\newblock \emph{IEEE Transactions on Signal Processing}, 65\penalty0 (13):\penalty0 3551--3582, 2017.

\bibitem[Vannieuwenhoven et~al.(2012)Vannieuwenhoven, Vandebril, and Meerbergen]{vannieuwenhoven2012new}
Nick Vannieuwenhoven, Raf Vandebril, and Karl Meerbergen.
\newblock A new truncation strategy for the higher-order singular value decomposition.
\newblock \emph{SIAM Journal on Scientific Computing}, 34\penalty0 (2):\penalty0 A1027--A1052, 2012.

\bibitem[Xu(2018)]{xu2018convergence}
Yangyang Xu.
\newblock On the convergence of higher-order orthogonal iteration.
\newblock \emph{Linear and Multilinear Algebra}, 66\penalty0 (11):\penalty0 2247--2265, 2018.

\bibitem[Zhang and Xia(2018)]{zhang2018tensor}
Anru Zhang and Dong Xia.
\newblock Tensor {SVD}: {S}tatistical and computational limits.
\newblock \emph{IEEE Transactions on Information Theory}, 64\penalty0 (11):\penalty0 7311--7338, 2018.

\end{thebibliography}

\newpage
\appendix

\section{Appendix}
\subsection{Proof of \Cref{lem:orthonormal_columns}}
\label{app:orthonormal_columns}

\OrthonormalColumns*

\begin{proof}
We first show that, without loss of generality, each factor matrix $\mat{A}^{(n)}$ can have orthonormal columns.
For any $n \in [N]$, let $\mat{A}^{(n)} = \mat{Q} \mat{R}$ be a thin QR decomposition
where $\mat{Q} \in \R^{I_n \times R_n}$ and $\mat{R} \in \R^{R_n \times R_n}$.
The mode-$n$ unfolding of Tucker decomposition
$\tensor{Y} = \tensor{G} \times_1 \mat{A}^{(1)} \times_2 \mat{A}^{(2)} \times_3 \dots \times_N \mat{A}^{(N)}$
can be written as
\begin{align}
\label{eqn:tucker_unfolding_formula}
    \mat{Y}_{(n)}
    &=
    \mat{A}^{(n)} \mat{G}_{(n)}
    \parens{
    \mat{A}^{(1)} \otimes \dots \otimes \mat{A}^{(n-1)} \otimes \mat{A}^{(n+1)} \otimes \dots \otimes \mat{A}^{(N)}
    }^\intercal \\
    &=
    \mat{Q} \parens{\mat{R} \mat{G}_{(n)}}
    \parens{
    \mat{A}^{(1)} \otimes \dots \otimes \mat{A}^{(n-1)} \otimes \mat{A}^{(n+1)} \otimes \dots \otimes \mat{A}^{(N)}
    }^\intercal, \notag
\end{align}
where $\otimes$ is the Kronecker product.
This demonstrates how matrix $\mat{R}$ can be absorbed into the core tensor.

We now assume that all the factor matrices $\mat{A}^{(n)}$ have orthonormal columns and are fixed.
To solve for the optimal core tensor $\tensor{G}$,
we follow \citep[Section 4.2]{kolda2009tensor}
and write the objective function in vectorized form:
\[
  \norm{
    \text{vec}(\tensor{X})
    -
    (\mat{A}^{(1)} \otimes \mat{A}^{(2)} \otimes \dots \otimes \mat{A}^{(N)})\text{vec}(\tensor{G})}_{2}^{2},
\]
where $\text{vec}(\tensor{X}) \in \R^{I_1 I_2 \cdots I_N}$ is the vector that vertically stacks the entries of $\tensor{X}$, ordered lexicographically by index.
This is an ordinary least-squares linear regression problem~\citep{fahrbach2022subquadratic}, so it follows from properties of the Kronecker product, pseudoinverse, and columnwise orthonormal matrices that
\begin{align*}
  \text{vec}(\tensor{G})
  &=
  (\mat{A}^{(1)} \otimes \mat{A}^{(2)} \otimes \dots \otimes \mat{A}^{(N)})^{+} \text{vec}(\tensor{X}) \\
  &=
  ({\mat{A}^{(1)}}^{+} \otimes {\mat{A}^{(2)}}^{+} \otimes \dots \otimes {\mat{A}^{(N)}}^{+}) \text{vec}(\tensor{X}) \\
  &=
  ({\mat{A}^{(1)}}^\intercal \otimes {\mat{A}^{(2)}}^\intercal \otimes \dots \otimes {\mat{A}^{(N)}}^\intercal) \text{vec}(\tensor{X}),
\end{align*}
which further implies that
$
  \tensor{G}
    = 
    \tensor{X} \times_{1} {\mat{A}^{(1)}}^\intercal \times_{2} {\mat{A}^{(2)}}^\intercal \times_{3} \dots \times_{N} {\mat{A}^{(N)}}^\intercal.
$
Applying this formula for $\tensor{G}$ and the $n$-mode product properties in \eqref{eqn:n_mode_product_same} to the definition of Tucker decomposition completes the proof.
\end{proof}

\subsection{Difference Between SVD and HOSVD for $N=2$}
\label{app:svd}

In this section, we explore the difference between the \emph{singular value decomposition} (SVD) and HOSVD for $N = 2$ in detail.
Recall that our construction in \Cref{sec:lower_bound} is
\begin{align*}
    \mat{X}
    =
    \mat{X}_{(1)}
    =
    \mat{X}_{(2)}
    =
    \begin{bmatrix}
        \sqrt{1 + \varepsilon} & 0 & 0 \\
        0 & 0 & 1 \\
        0 & 1 & 0
    \end{bmatrix}.
\end{align*}
The SVD of $\mat{X} = \mat{U} \mat{\Sigma} \mat{V}^\intercal$ can be written as:
\begin{align*}
    \mat{U}
    &=
    \begin{bmatrix}
        1 & 0 & 0 \\
        0 & 1 & 0 \\
        0 & 0 & 1 \\
    \end{bmatrix} ,\\
    \mat{\Sigma}
    &=
    \begin{bmatrix}
        \sqrt{1 + \varepsilon} & 0 &  0 \\
        0 & 1 & 0 \\
        0 & 0 & 1 \\
    \end{bmatrix}, \\
    \mat{V}
    &=
    \begin{bmatrix}
        1 & 0 & 0 \\
        0 & 0 & 1 \\
        0 & 1 & 0 \\
    \end{bmatrix}.
\end{align*}
Note that the SVD is not unique since $\sigma_2 = \sigma_3$, which we can see by writing
$\mat{X} = \mat{X}^\intercal = \mat{V} \mat{\Sigma} \mat{U}^\intercal$.

\paragraph{HOSVD.}
Since $\mat{X}$ is symmetric and $\mat{r} = (2,2)$,
\Cref{fact:symmetry} implies that
\begin{align*}
    \mat{A}^{(1)}
    =
    \mat{A}^{(2)}
    =
    \begin{bmatrix}
        1 & 0 \\
        0 & 1 \\
        0 & 0 \\
    \end{bmatrix}
    \text{~~or~~}
    \mat{A}^{(1)}
    =
    \mat{A}^{(2)}
    =
    \begin{bmatrix}
        1 & 0 \\
        0 & 0 \\
        0 & 1 \\
    \end{bmatrix}.
\end{align*}
Applying the mode-$n$ unfolding formula for Tucker decompositions in~\eqref{eqn:tucker_unfolding_formula}
to the the optimal core tensor in \Cref{lem:orthonormal_columns}
gives us
\begin{align*}
    \mat{G}
    =
    {\mat{A}^{(2)}}^\intercal \mat{X} {\mat{A}^{(1)}}
    =
    \begin{bmatrix}
        \sqrt{1 + \varepsilon} & 0 \\
        0 & 0
    \end{bmatrix},
\end{align*}
in both cases.
It follows that the HOSVD reconstruction is always the top component (\Cref{lem:hosvd_error}):
\begin{align*}
    \widehat{\mat{X}}_{\HOSVD(\mat{r})}
    &=
    \mat{A}^{(2)} \mat{G} {\mat{A}^{(1)}}^\intercal
    =
    \begin{bmatrix}
        \sqrt{1 + \varepsilon} & 0 & 0 \\
        0 & 0 & 0 \\
        0 & 0 & 0
    \end{bmatrix}.
\end{align*}
This gives a reconstruction error of $\norm{\widehat{\mat{X}}_{\HOSVD(\mat{r})} - \mat{X}}_{\frobenius}^2 = 2$.

\paragraph{SVD.}
If we take the rank-2 truncated SVD of $\mat{X}$ as the Tucker decomposition,
then depending on our choice of $\mat{U}$,
\begin{align*}
    \widehat{\mat{X}}_{\text{SVD}(2)}
    &=
    \begin{bmatrix}
        1 & 0 \\
        0 & 1 \\
        0 & 0 \\
    \end{bmatrix}
    \begin{bmatrix}
        \sqrt{1 + \varepsilon} & 0 \\
        0 & 1 \\
    \end{bmatrix}
    \begin{bmatrix}
        1 & 0 \\
        0 & 0 \\
        0 & 1 \\
    \end{bmatrix}^\intercal
    =
    \begin{bmatrix}
        \sqrt{1 + \varepsilon} & 0 & 0 \\
        0 & 0 & 1 \\
        0 & 0 & 0
    \end{bmatrix}
\end{align*}
or
\begin{align*}
    \widehat{\mat{X}}_{\text{SVD}(2)}
    &=
    \begin{bmatrix}
        1 & 0 \\
        0 & 0 \\
        0 & 1 \\
    \end{bmatrix}
    \begin{bmatrix}
        \sqrt{1 + \varepsilon} & 0 \\
        0 & 1 \\
    \end{bmatrix}
    \begin{bmatrix}
        1 & 0 \\
        0 & 1 \\
        0 & 0 \\
    \end{bmatrix}^\intercal
    =
    \begin{bmatrix}
        \sqrt{1 + \varepsilon} & 0 & 0 \\
        0 & 0 & 0 \\
        0 & 1 & 0
    \end{bmatrix}.
\end{align*}
The reconstruction error is $\norm{\widehat{\mat{X}}_{\text{SVD}(2)} - \mat{X}}_{\frobenius}^2 = 1$ in both cases, so the approximation factor in \Cref{thm:hosvd} is tight.

The key difference between SVD and HOSVD is the \emph{asymmetry} in the top-2 left and right singular vectors $\mat{U}[:,:2]$ and $\mat{V}[:,:2]$.
This motivated our construction of the adversarial symmetric tensor that exploits the greedy behavior of HOSVD in each step.

\end{document}